\documentclass[12pt,technote, onecolumn, draft]{IEEEtran}
\usepackage[utf8]{inputenc}
\usepackage{latexsym}
\usepackage{mathrsfs}
\usepackage{graphicx}
\usepackage{multirow}
\usepackage{rotating}
\usepackage{makecell}
\usepackage{amsfonts,amssymb,amsmath,amsthm,bm}
\usepackage{color}
\usepackage{booktabs}
\usepackage{cite}
\newtheorem{theorem}{Theorem}[section] 
\newtheorem{definition}[theorem]{Definition} 
\newtheorem{lemma}[theorem]{Lemma} 
\newtheorem{corollary}[theorem]{Corollary}
\newtheorem{example}[theorem]{Example}
\newtheorem{proposition}[theorem]{Proposition}
\newtheorem{remark}[theorem]{Remark}

\begin{document}    
	\title{Construction of Multi-Sequences With High Nonlinear Complexity via Narrow Ray Class Fields\thanks{Xiaofeng Liu and Fang-Wei Fu were supported by the National Key Research and Development Program of China (Grant Nos.  2022YFA1005000), the National Natural Science Foundation of China (Grant Nos. 12141108, 61971243), the Fundamental Research Funds for the Central Universities of China (Nankai University), and the Nankai Zhide Foundation. Jun Zhang was supported by the National Natural Science Foundation of China under Grant No. 12441105.}}
	\author{Xiaofeng Liu, Jun Zhang, Fang-Wei Fu
		\IEEEcompsocitemizethanks{\IEEEcompsocthanksitem Xiaofeng Liu and Fang-Wei Fu are with the Chern Institute of Mathematics and LPMC, Nankai University, Tianjin 300071, China, Emails: lxfhah@mail.nankai.edu.cn, fwfu@nankai.edu.cn. Jun Zhang is with the Center for Discrete Mathematics, and College of Mathematics and Statistics, Chongqing University, Chongqing 401331, China. Email: jun\_zhang@cqu.edu.cn.
		}
	}{\tiny }
	\maketitle	
	\begin{abstract}
		Nonlinear complexity is a fundamental criterion in the evaluation of pseudorandom sequences. The construction of multi-sequences with high nonlinear complexity is both theoretically and practically important in cryptography. Motivated by prior constructions of multi-sequences with high nonlinear complexity in [IEEE Trans. Inf. Theory, 60(10), 2014] and [IEEE Trans. Inf. Theory, 63(12), 2017], we provide a unified framework via narrow ray class fields and the cyclic descent introduced by Guruswami and Xing in [J. Combin. Theory Ser. A 129 (2015) ]. Then we can generate new multi-sequences with high nonlinear complexity over function fields with arbitrary genera.
        \end{abstract}
	\begin{IEEEkeywords}
		 Nonlinear complexity, multi-sequences, $\beta-$error nonlinear complexity, narrow ray class fields.
	\end{IEEEkeywords}
	
	\section{Introduction}
	Nonlinear complexity measures the shortest nonlinear feedback rule capable of generating a sequence. The constructions of multi-sequences with high nonlinear complexity are important in  simulation and word-based ciphers. In general, a multi-sequence may have high linear complexity, while still having low nonlinear complexity. Therefore, developing a systematic method to generate multi-sequences with high nonlinear complexity is of both theoretical interest and practical importance.

    Function fields have been widely applied in coding theory and cryptography.  In particular, function fields serve as a powerful tool for generating pseudorandom sequences. Prior research has extensively examined the linear complexity and correlation of sequences constructed via function fields. However, the construction of multi-sequences with high nonlinear complexity via function fields remains comparatively limited.
    
    Now we recall some known constructions of multi-sequences. In 2009, Xing and Ding constructed multi-sequences with large linear complexity and $k-$error linear complexity from Hermitian function fields; see \cite{1}. In 2014, Niederreiter and Xing constructed a family of multi-sequences with high nonlinear complexity via the Hermitian function field; see \cite{4}. 
    In 2017, Luo \textit{et al.} constructed a family of sequences with high nonlinear complexity via cyclotomic function fields; see \cite{3}. In 2023, $\ddot{\mathrm{O}}$zbudak and Tutas constructed two families of multi-sequences with high nonlinear complexity via Hermitian and Suzuki function fields; see \cite{15}. In 2023, Castellanos \textit{et al.} designed a family of multi-sequences from the generalized Hermitian function field; see \cite{21}. 
    In 2026, Hu and Ma studied the automorphism groups of a special maximal function field, known as the Abdon-Torres function field; see \cite{29}. Recently, Hu \textit{et. al.} also proposed a general construction of multi-sequences and provided some explicit constructions of multi-sequences with high linear complexity; see \cite{100}. 
   
    Narrow ray class fields also serve as a powerful tool in coding theory. In 1997, Niederreiter and Xing proposed a general method to construct function fields with many rational places via Drinfeld modules of rank $1$ and narrow ray class fields; see \cite{49}. In 2015, Guruswami and Xing achieved optimal rate algebraic list decoding from descent function fields via narrow ray class fields; see \cite{26}.
	It turns out that such descent function fields are also ideal to generate multi-sequences. In Section \ref{sec:3}, we proved that the multi-sequences constructed via narrow ray class fields have high dimensions, large linear and nonlinear complexity. In particular, constructions in \cite{2,3} can be generalized in narrow ray class field constructions.

	The rest of this paper is organized as follows. In Section \ref{sec:2}, we provide some preliminaries on nonlinear complexity, fundamental results of function fields, class field theory, and Ihara's constant. In Section \ref{sec:7}, we list the main results and techniques used throughout this paper.
	In Section \ref{sec:3}, we provide a construction of multi-sequences via narrow ray class fields. It turns out that such multi-sequences achieve favorable cryptographic properties.  We also provide some explicit constructions to better illustrate our construction.  In Section \ref{sec:5}, we provide some asymptotic results via function field towers.
	Finally, in Section \ref{sec:6}, we provide concluding remarks and outline some directions for future research on nonlinear complexity and constructions of multi-sequences.

	\section{Preliminaries}
	\label{sec:2}
	\subsection{The Definition and Some Indicators of multi-sequences}
    \subsubsection{ Definition of multi-sequences}
    \begin{definition}
		Let $\mathbb{F}_{q}$ be a finite field with $q$ elements. A periodic multi-sequence of dimension $b$ and period $e$ is defined as a family 
		\begin{displaymath}
			\mathcal{S}=\{\mathbf{s}^{(t)}: 1\leq t\leq b\},\ \mathbf{s}^{(t)}=\{s_{i}^{(t)}\}^{e}_{i=1}.
		\end{displaymath} 
	\end{definition}
	Several multi-sequences with high nonlinear complexity have been constructed via different methods over the past decades. Now we provide two classical constructions of multi-sequences via function fields to better illustrate the definition.
	\begin{example}
		\begin{enumerate}
			\item The multi-sequence $\mathcal{S}_{1}=\{\mathbf{s}^{1,t}: 1\leq t\leq \sqrt{q}-1\}$ constructed via the Hermitian function field $F/\mathbb{F}_{\sqrt{q}}$ defined by $\mathcal{Z}:y^{\sqrt{q}}+y=x^{\sqrt{q}+1}$, where each component sequence is given by $\mathbf{s}^{1,t}=\{z(\sigma^{i}(P_{t}))\}_{i=1}^{q-1}, 1\leq t\leq \sqrt{q}-1$; see \cite{2}.
			\item The multi-sequence $\mathcal{S}_{2}=\{\mathbf{s}^{2,t}: 1\leq t\leq \sqrt{q}-1\}$ constructed via cyclotomic function field, where each component sequence is given by $\mathbf{s}^{2,t}=\{z(\sigma^{i}(P_{t}))\}_{i=0}^{\frac{\sqrt{q}^{d}+1}{\sqrt{q}+1}-1}, 1\leq t\leq \sqrt{q}$; see \cite{3}.
		\end{enumerate}
	\end{example}
	\subsubsection{Linear Complexity}
	\begin{definition}
		The linear complexity of a multi-sequence $\mathcal{S}=\{\mathbf{s}^{(t)}:1\leq t\leq b\}$ is defined as the smallest nonnegative integer $\mathrm{LC}(\mathcal{S}):=\ell$ such that there exist elements $\lambda_0,\lambda_1,\dots,\lambda_{\ell-1}\in\mathbb{F}_q$ satisfying:
		\begin{displaymath}
			s^{(t)}_{j+\ell}=\sum^{\ell-1}_{i=0}\lambda_{i} s^{(t)}_{j+i},
		\end{displaymath}
		for all component sequences $\mathbf{s}^{(t)}$ and any index $j$. Note that each multi-sequence with period $e$ satisfies $\mathrm{LC}(\mathcal{S})\leq e$.
	\end{definition}
	
	\subsubsection{ Nonlinear Complexity}
	\begin{definition}\label{d1}
		Fix an integer $m\geq 1$, the $m$th nonlinear complexity (nonlinear complexity for short) of the multi-sequence $\mathcal{S}=\{\mathbf{s}^{(t)}:1\leq t\leq b\}$ is defined as the least integer $ \mathrm{NL}_{m}(\mathcal{S}):=\mathfrak{N}\geq 1$ such that there exists at least one multivariate polynomial $\Phi(x_{0},\cdots,x_{ \mathfrak{N}-1})\in\mathbb{F}_{q}[x_{0},\cdots,x_{\mathfrak{N}-1}]$ of the total degree at most $m$ such that
        \begin{displaymath}
s_{j+\mathfrak{N}}^{(t)}=\Phi(s_{j}^{(t)},\cdots,s_{j+\mathfrak{N}-1}^{(t)})
        \end{displaymath}
        for any $1\leq t\leq b$ and  $0\leq j\leq e-\mathfrak{N}-1$. Note that $\mathrm{NL}_{m}(\mathcal{S})\leq\mathrm{LC}(\mathcal{S})$.
	\end{definition}

	\begin{remark}
     Denote by  $\mathrm{NL}_{m}(\mathcal{S})$ the  nonlinear complexity.   If we replace the condition ``of total degree at most $m$" with  ``of degree at most $m$ in each variable" in the definition \ref{d1}, then we have the definition of joint nonlinear complexity $\mathrm{NE}_{m}(\mathcal{S})$ in \cite{5}. It is also obvious that $\mathrm{NL}_{m}(\mathcal{S})\geq \mathrm{NE}_{m}(\mathcal{S})$; see \cite{10}.
	\end{remark}

	\subsubsection{ $\beta-$error Nonlinear Complexity }
	Fix a positive integer $\beta\geq 1$. Now we shall define the $\beta-$error nonlinear complexity as follows.  It is defined as the minimal nonlinear complexity of $\mathcal{S}$ by changing in total at most $\beta$ terms for all component sequences in $\mathcal{S}$. The formal definition is given below.
	\begin{definition}
		Given two multi-sequences $\mathcal{S,U}$ with period $e$ and dimension $b$, let
        \begin{displaymath}
            d_{H}(\mathcal{S,U})=\sum^{b}_{t=1}d_{H}(\mathbf{s}^{(t)},\mathbf{u}^{(t)})
        \end{displaymath}
        be the total Hamming distance between $\mathcal{S}$ and $\mathcal{U}$. The $\beta-$error nonlinear complexity is defined as
        \begin{displaymath}
            \mathrm{NL}_{m,\beta}(\mathcal{S})=\min_{d_{H}(\mathcal{S,U})\leq\beta}\mathrm{NL}_{m}(\mathcal{U}).
        \end{displaymath}
	\end{definition}

	\subsection{Extension Theory of Function Fields}
	Consider an algebraic function field $F/\mathbb{F}_{q}$ with genus $\mathfrak{g}_{F}$. A place $P$ is the maximal ideal of some valuation ring $\mathcal{O}_{P}$ of $F/\mathbb{F}_{q}$. The degree of $P$, denoted by $\deg(P)$, is determined by the extension $F_{P} = \mathcal{O}_{P}/P$, \textit{i.e.}, $\deg(P) = [F_P : \mathbb{F}_{q}]$. A place $P$ with $\deg(P) = 1$ is called rational. The number of rational places of $F$ is denoted by $N(F)$. According to the Serre bound, we have $|N(F)-q-1|\leq \mathfrak{g}_{F}\cdot\lfloor2\sqrt{q}\rfloor,$
	where $\lfloor x\rfloor$ is denoted by the greatest integer less than or equal to $x$. A divisor is formally defined as a finite linear combination of places.   Denote by $\mathbb{P}_{F}$ and $\mathbb{D}_{F}$ the collection of all places and divisors, respectively. In addition, a normalized discrete valuation $\mathrm{v}_{P}$ is assigned to each place $P$ of $F$. 
	
	For any divisor $G \in \mathbb{D}_{F}$, we denote by $\mathrm{v}_{P}(G)$ the coefficient of $P$ in $G$. The support of a divisor $G$ is defined as
	$\mathrm{Supp}(G) = \{ P \in \mathbb{P}_{F} \mid \mathrm{v}_{P}(G) \neq 0 \}.$
	The degree of a divisor $G$ is given by $\deg(G) = \sum_{P \in \mathbb{P}_{F}} \mathrm{v}_{P}(G)\deg(P).$
	A divisor $G = \sum_{P \in \mathbb{P}_{F}} \mathrm{v}_{P}(G)P$ is called effective if $\mathrm{v}_{P}(G) \ge 0$ for each $P \in \mathbb{P}_{F}$.
	
	 For any non-negative divisor $G$, the associated Riemann–Roch space is defined as
	$\mathcal{L}(G):=\{f\in F\setminus\{0\}\mid (f)+G\geq 0\}\cup\{0\}$.  If $\deg(G)\geq 2\mathfrak{g}_{F}-1$, then we have $\dim\mathcal{L}(G)=\deg(G)+1-\mathfrak{g}_{F}$ from the Riemann-Roch theorem.
	
	The automorphism group $\mathrm{Aut}(F/\mathbb{F}_{q})$ of  $F$ over $\mathbb{F}_{q}$ is defined by
	\[
	\mathrm{Aut}(F/\mathbb{F}_{q})=\{\sigma : F\rightarrow F \mid \sigma \text{ is an }\mathbb{F}_{q}\text{-automorphism of }F\}.
	\]
	For any rational place $P$ of $F$ and any automorphism $\sigma\in \mathrm{Aut}(F/\mathbb{F}_{q})$, the image $\sigma(P)$ is again a place of $F$, and the following properties hold.
	
	\begin{lemma}(see \cite{11})
		Given any rational place $P$, $f\in F$ and $\sigma\in \mathrm{Aut}(F/\mathbb{F}_{q})$, then we have:
		\begin{enumerate}
			\item  $\deg(\sigma(P)) =\deg(P)$;
			\item  $\mathrm{v}_{\sigma(P)}(\sigma(f))=\mathrm{v}_{P}(f)$;
			\item $\sigma(f)(\sigma(P)) =f(P)$, if $\mathrm{v}_{P}(f)\geq0$.
		\end{enumerate}
		\label{l1}
	\end{lemma}
	Let $H/\mathbb{F}_{q}$ be a finite algebraic extension of $F$. 
	For any place $P\in\mathbb{P}_{F}$ and $Q\in\mathbb{P}_{H}$ with $Q|P$, we denote by $d(Q|P), e(Q|P), f(Q|P)$ different exponent, ramification index, and relative degree of $Q|P$, respectively. The different of $H/F$ is defined as $\mathrm{Diff}(H/F)=\sum_{P\in\mathbb{P}_{F}}\sum_{Q|P}d(Q|P)Q$. The different exponent $d(Q\mid P)$ satisfies $d(Q|P)= e(Q|P)-1$ if and only if  $\gcd(e(Q|P),p)=1$. Otherwise, the different exponent can be calculated by 
	\begin{displaymath}
		d(Q|P)=\sum^{\infty}_{i=0}(|G_{i}(Q|P)|-1),
	\end{displaymath}
	where $  |G_{i}(Q|P)|$ is denoted as orders of higher ramification groups $G_{i}(Q|P)=\{\sigma\in\mathrm{Gal}(H/F): v_{Q}(\sigma(z)-z)\geq i+1\ \text{for all}\ z\in H\},$ for any $ i\geq 0$, see~\cite{7}. The Hurwitz genus formula (see \cite{7}) is given by
	\begin{displaymath}
		2\mathfrak{g}_{H}-2=[H: F]\cdot (2\mathfrak{g}_{F}-2)+\deg\mathrm{Diff}(H/F).
	\end{displaymath}

	\subsection{ Drinfeld Modules of Rank $1$ and Class Field Theory }
	In this subsection, we choose $F/\mathbb{F}_{q}$ as any function field with  $\mathfrak{g}_{F}\geq 1$. Suppose $N(F)\geq 1$ and fix a rational place $\infty$ of $F$. Denote by $A_{q}$ the integral ring $A_{q}=\{x\in F : \mathrm{v}_{P}(x)\geq 0\ \text{for all}\ P\in\mathbb{P}_{F}\ \text{with}\ P\neq\infty\}$. The Hilbert class field of $A_{q}$, denoted by $\mathcal{H}_{A_{q}}$, is the maximal abelian extension such that any finite place of $F$ is unramified and $\infty$ splits completely in $H_{A_{q}}$. We have $[\mathcal{H}_{A_{q}}:F]=\mathrm{h}(A_{q})=\mathrm{h}(F)$; see \cite{4}. 
	
	Let $\tau : c\mapsto c^{q}$ be the Frobenius endomorphism of $\mathcal{H}_{A_{q}}$. The left twisted polynomial ring $\mathcal{H}_{A_{q}}\{\tau\}$ is defined by $\tau u=u^{q}\tau$ for all $u\in \mathcal{H}_{A_{q}}$.  A Drinfeld module of rank $1$ is defined as a homomorphism $\psi:A_{q}\to \mathcal{H}_{A_{q}}\{\tau\}, a\mapsto \psi_{a}$ such that $\deg_{\tau}(\psi_{a})=-\mathrm{v}_{\infty}(a)$.
	
	Denote by $F_{\infty}$ the $\infty$-adic completion of $F/\mathbb{F}_{q}$. The sign function is defined such that $\mathrm{sgn}:F^{*}_{\infty}\to\mathbb{F}_{q}^{*}$ with
	$\mathrm{sgn}(\alpha)=\alpha$ for each $\alpha\in\mathbb{F}_{q}^{*}$. It also satisfies $\mathrm{sgn}(U^{(1)}_{F_{\infty}})=1$ for $U^{(1)}_{F_{\infty}}=\{x\in U_{F_{\infty}} : \mathrm{v}_{\infty}(x-1)\geq 1\}$.
	
	Let $\overline{\mathcal{H}_{A_{q}}}$ be the fixed algebraic closure of $\mathcal{H}_{A_{q}}$ and let $\psi$ be a sgn-normalized Drinfeld module of rank $1$ over $\mathcal{H}_{A_{q}}$. Let $M$ be a nonzero ideal of $A_{q}$, the $M$-torsion module $\Lambda_{\psi}(M)$ associated with $\psi$ is defined by
	\begin{displaymath}
		\Lambda_{\psi}(M)=\left\{\, t\in \bigl(\overline{\mathcal{H}_{A_{q}}},+\bigr)\;:\;\psi_{M}(t)=0 \right\}.
	\end{displaymath}
	If $\gcd(M,\mathrm{char}_{A_{q}}(\psi))=1$, then $\Lambda_{\psi}(M)$ is a finite set with
	$|\Lambda_{\psi}(M)|=q^{\deg(\psi_{M})}$. 
	
	The elements in $\Lambda_{\psi}(M)$ are called the $M$-torsion elements in $\left(\overline{\mathcal{H}_{A_{q}}},+\right)$. For any nonzero ideal $M$ of $A$, we associate with $M$ the divisor
	$D :=\sum_{P\in\infty'} \mathrm{v}_{P}(M)P \in \mathbb{D}_{F}$. Then we have $\Lambda_{\psi}(M)=\Lambda_{\psi}(D)$ and the class field $\mathcal{H}_{A_{q}}(\Lambda_{\psi}(D))$, which is called the $M-$division field of $\psi$. Denote the Euler function of $P\in\mathbb{P}_{F}$ by $\Phi(P)=|(A_{q}/P)^{*}|$. We have $\mathrm{Gal}(\mathcal{H}_{A_{q}}(\Lambda_{D})/F)\simeq \mathrm{Cl}^{+}_{D}(A_{q})$ and the cardinality $ |\mathrm{Cl}^{+}_{D}(A_{q})|=\mathrm{h}(F)(q^{d}-1) $ if the degree of $P$ is taken as $\deg(P)=d$; see Corollary 2.6 in \cite{6}.

    Now we summarize the prime ideal decompositions in $\mathcal{H}_{A_{q}}(\Lambda_{M})$ from Proposition 3.3.8 in \cite{4}.
	\begin{proposition} (Proposition 3.3.8, \cite{4})\label{p1}
		Let $\mathcal{H}_{A_{q}}(\Lambda_{Q})/\mathcal{H}_{A_{q}}$ be a class field with a prime modulus $Q$. Denote by $\mathcal{O}_{\mathcal{H}_{A_{q}}(\Lambda_{Q})}$ the integral ring of $\mathcal{H}_{A_{q}}(\Lambda_{Q})/\mathcal{H}_{A_{q}}$. Then the following statements hold.
		
         $1)$ Any finite place $R\nmid Q$ is unramified. Suppose $f$ is the order of $R$ modulo $Q$, \textit{i.e.}, $R^{f}\equiv 1$ modulo $Q$, then 
			$$R\mathcal{O}_{\mathcal{H}_{A_{q}}(\Lambda_{P})}=\mathfrak{Q}_{1}\mathfrak{Q}_{2}\cdots\mathfrak{Q}_{g}, e=1, f_{i}=f,g=\frac{\Phi(Q)}{f}=\frac{q^{d}-1}{f}.$$
			Then $R$ splits in $\mathcal{H}_{A_{q}}(\Lambda_{Q})$ completely if and only if $f=1$ and $R$ remains inert if and only if $f=\Phi(Q)$.
			
             $2)$ $Q$ is totally ramified in $\mathcal{H}_{A_{q}}(\Lambda_{Q})/\mathcal{H}_{A_{q}}$ and $Q\mathcal{O}_{\mathcal{H}_{A_{q}}(\Lambda_{Q})}=\mathfrak{P}^{\Phi(Q)}$ with the ramification index given by $e_{Q}( \mathcal{H}_{A_{q}}(\Lambda_{Q})/\mathcal{H}_{A_{q}} )=\Phi(Q)$, $f=g=1$.
			
            $3)$ The infinite place $\infty$ is ramified in $\mathcal{H}_{A}(\Lambda_{Q})/\mathcal{H}_{A_{q}},$ and ramification index satisfies 
          $$ e_{\infty}(\mathcal{H}_{A_{q}}(\Lambda_{Q})/\mathcal{H}_{A_{q}})=q-1,f=1, g=\Phi(Q)/(q-1)=(q^{d}-1)/(q-1).$$
	\end{proposition}
	
	\subsection{Ihara's Constant $A(q)$}
	Let $\mathfrak{g}$ be a positive integer and we denote by $N_{q}(\mathfrak{g})$ the maximum number of rational places of any function field with constant field $\mathbb{F}_{q}$ and genus $\mathfrak{g}$. Define the real number 
	\begin{displaymath}
		A(q)=\lim\sup_{\mathfrak{g}\to\infty}\frac{N_{q}(\mathfrak{g})}{\mathfrak{g}}
	\end{displaymath}
	and it is called Ihara's constant. By Drinfeld-Vl{\u{a}}du{\c{t}} bound, the following upper bound holds
	\begin{displaymath}
		A(q)\leq\sqrt{q}-1;
	\end{displaymath}
	see \cite{7}. In particular, if $q$ is a square, then the equality holds
	\begin{displaymath}
		A(q)=\sqrt{q}-1
	\end{displaymath}
and it is proved via modular curves; see \cite{30}.
    	\subsection{Garcia-Stichtenoth Tower}
	Let $q$ be a square of a prime power. The well-known Garcia-Stichtenoth tower $\mathcal{T}=(T_{1},T_{2},\cdots)$ with the initial function field $T_{1}=\mathbb{F}_{q}(y_{1})$. The iterative process $T_{m}=T_{m-1}(y_{m})$ is with the defining equation
	\begin{displaymath}
		y_{m}^{\sqrt{q}}+y_{m}=\frac{y^{ \sqrt{q} }_{m-1}}{y^{  \sqrt{q}  -1}_{m-1}+1}
	\end{displaymath}
	for $m\geq 2$. Now we summarize some fundamental properties of Garcia-Stichtenoth towers to the following proposition:
	\begin{proposition}(Section 7.4, \cite{7})
		\begin{enumerate}
			\item Let $\infty$ be the unique pole of $y_{1}$ in $T_{1}$. Then $\infty$ is totally ramified in $T_{m}/T_{1}$ for each $m\geq 2$. Let $P_{\infty,m}$ be the unique place of $T_{m}$ lying over $\infty$. Then $P_{\infty,m}$ is a common place of $y_{1},y_{2},\cdots,y_{m}$.
			\item Let $Q_{\alpha}$ denote the zero of $y_{1}-\alpha$ in $T_{1}$. Then any $Q_{\alpha}$ with $\alpha^{ \sqrt{q}   -1}=-1$ is totally ramified in $T_{m}/T_{1}$. Any rational place $Q_{\alpha}$ with $\alpha^{  \sqrt{q}  -1}\neq -1$ splits completely in $T_{m}/T_{1}$ and then the number of rational points satisfies $N(T_{m})\geq (q- \sqrt{q}   ) \sqrt{q}  ^{m-1}+ \sqrt{q}  $.
			\item The genus of $T_{m}$ is given by
			\begin{displaymath}
				\mathfrak{g}_{T_{m}}=\begin{cases}
					(  \sqrt{q} ^{\frac{m}{2}}-1)^{2},&\text{if}\ m\equiv 0\ (\text{mod}\ 2),\\
					(  \sqrt{q}  ^{\frac{m+1}{2}}-1)( \sqrt{q}   ^{\frac{m-1}{2}}-1),&\text{if}\ m\equiv 1\ (\text{mod}\ 2).
				\end{cases}
			\end{displaymath}
			\item $\lim_{m\to\infty}\frac{N(T_{m})}{\mathfrak{g}_{T_{m}}}=\sqrt{q}-1$ and $\mathcal{T}$ is asymptotically optimal.
		\end{enumerate}
	\end{proposition}

   \section{Main Results and comparison}
   \label{sec:7}
	In this section, we briefly introduce the main results of this paper.
	Let $\mathbb{F}_{q}$ be a finite field with $ q=r^{2}$. For a given function field $F/\mathbb{F}_{r}$ with genus $\mathfrak{g}_{F}$, we can construct a family of multi-sequences with high linear and nonlinear complexity.
	
	Let $L=F\cdot\mathbb{F}_{q}$ denote a constant field extension and fix a place $Q$ of odd degree $d$. Define two integral rings
	\begin{displaymath}
		\begin{split}
			A_{q}&=\{x\in F/\mathbb{F}_{q}: \mathrm{v}_{P}(x)\geq 0\ \text{for all}\ P\in\mathbb{P}_{F/\mathbb{F}_{q}}\ \text{with}\ P\neq\infty\}\\
			A_{r}&=\{x\in F/\mathbb{F}_{r}: \mathrm{v}_{P}(x)\geq 0\ \text{for all}\ P\in\mathbb{P}_{F/\mathbb{F}_{r}}\ \text{with}\ P\neq\infty\}.
		\end{split}
	\end{displaymath}
	
	Let $\mathcal{H}_{A_{q}}(\Lambda_{Q})$ be the $Q-$division field and $\mathrm{Cl}^{+}_{Q}(A_{r}) $ be a narrow ray class group modulo $Q$ of $F/\mathbb{F}_{r}$.
	Let $\Gamma\subseteq \mathcal{H}_{A_{q}}(\Lambda_{Q})$ be the first descent subfield fixed by the subgroup $\mathbb{F}_{q}^{*}\cdot\mathrm{Cl}^{+}_{Q}(A_{r})$, \textit{i.e.}, $\Gamma= \mathcal{H}_{A_{q}}(\Lambda_{Q})^{\mathbb{F}_{q}^{*}\cdot\mathrm{Cl}^{+}_{Q}(A_{r})}  $. The second descent subfield is denoted by $H$ and the degree is denoted by $$[H:L]=\frac{r^{d}+1}{r+1}:=\mu_{H}.$$

    Denote by $\mathcal{R}_{Q}=\sum_{R\mid Q}R$ the reduced divisor supported on the place of $H$ lying above $Q$.    Choose a divisor $A=\gamma\mathcal{R}_{Q}$ with $ a:=\deg A=\gamma\Delta_{Q}$ and $\gamma\in\mathbb{N}$. Then we have the following results.
	\begin{theorem}   
	 let $\mathbb{F}_{q}$ with $ q=r^{2}$ be an extension of finite field $\mathbb{F}_{r}$.    Consider a function field $F/\mathbb{F}_{r}$ with genus $\mathfrak{g}_{F}$ and $1+r+s_{r}, |s_{r}|\leq \mathfrak{g}_{F}\lfloor 2\sqrt{r}\rfloor$ rational places. Fix a place $Q\in\mathbb{P}_{F}$ of odd degree $d$ and $ 2\mathfrak{g}_{H}-1  \leq a\leq (r+s_{r}-1) \mu_{H} -1$.  Then we can construct a multi-sequence $$\mathcal{S}=\{\mathbf{s}^{(t)}:1\leq t\leq r+s_{r}\}$$ with dimension $r+s_{r}$, period $\mu_{H} $, linear complexity
		\begin{displaymath}
			\mathrm{L}(\mathcal{S})= \mu_{H},
		\end{displaymath}
        and nonlinear complexity lower bounded by
		\begin{displaymath}
			\mathrm{NL}_{m}(\mathcal{S})\geq\frac{(r+s_{r})\mu_{H}-ma-1}{m+r+s_{r}}.
		\end{displaymath}

	\end{theorem}
	In particular, we study $\beta-$error nonlinear complexity. Let $\alpha_{0}=1+\lceil (a+1)/\mu_{H}\rceil$, where $\lceil x\rceil$ denotes the ceiling of $x\in\mathbb{R}$, then we have the following results.
	\begin{theorem}
    Let $\mathbb{F}_{r}$ be a finite field with $r$ elements and let $\mathbb{F}_{q}, q=r^{2}$ be an extension of $\mathbb{F}_{r}$.    Consider a function field $F/\mathbb{F}_{r}$ with genus $\mathfrak{g}_{F}$ and $1+r+s_{r}, |s_{r}|\leq \mathfrak{g}_{F}\lfloor 2\sqrt{r}\rfloor$ rational places. Fix a place $Q\in\mathbb{P}_{F}$ of odd degree $d$ and $2\mathfrak{g}_{H}-1  \leq a\leq (r+s_{r}-1) \mu_{H} +\alpha_{0}-3$.   Then we can construct a multi-sequence $$\mathcal{S}=\{\mathbf{s}^{(t)}:1\leq t\leq r+s_{r}\}$$ with dimension $r+s_{r}$, period $\mu_{H}$, and $\beta-$error nonlinear complexity is divided into the following cases:
		
		$1)$ If $\beta\leq r+s_{r}-\alpha_{0}$,
		\begin{displaymath}
			\mathrm{NL}_{m,\beta}(\mathcal{S})\geq \frac{\mu_{H}(r+s_{r}-\beta)-ma-1}{m+ r+s_{r}-\beta   };
		\end{displaymath}

		$2)$ If $\beta\leq r+s_{r}-\alpha_{0}+1$,   $$  \mathrm{NL}_{m,\beta}(\mathcal{S})\geq   -1+\frac{\mu_{H}(r+s_{r})-ma+m-1}{ r+s_{r}-\alpha_{0}+1+m   };$$
		
		$3)$ If $\beta\leq r+s_{r}-\alpha_{0}+2$,  $$  \mathrm{NL}_{m,\beta}(\mathcal{S})\geq  -1+\min\left\{\frac{\mu_{H}(r+s_{r}-1)-ma+m-1}{ r+s_{r}-\alpha_{0}+m   },    \frac{\mu_{H}(r+s_{r})-ma+m-1}{ r+s_{r}-\alpha_{0}+2+m   }  \right\}.  $$
	\end{theorem}

In the second part of this section, we make a detailed comparison with prior constructions of multi-sequences with high linear complexity and nonlinear complexity.  Note that most of the constructions are based on large automorphism groups. However, algebraic curves with large automorphism groups remain limited. For an algebraic curve $X/\mathbb{F}_{q}$ of genus $\mathfrak{g}_{X}$, the Hurwitz inequality states that for any irreducible algebraic curve $X$, $|\mathrm{Aut}(X/\mathbb{F}_{q})|\leq 16\mathfrak{g}_{X}^{4}$; see \cite{30}.   It turns out that our construction can generate new infinite families of multi-sequences with a large period and high dimension. For detailed comparison, several constructions are listed in table \ref{q}.

	\begin{table}[htbp]
		\centering
		\caption{Parameters of multi-sequence families via function fields} 
		\resizebox{1\textwidth}{!}{
			\begin{tabular}{|c|c|c|c|}  
				\toprule
				Function Field  & Dimension  &Linear complexity & Nonlinear complexity \\    \midrule
                Rational Function Field (\cite{5}) &  $n\geq 1$ &  $ q-1$  & \makecell{$  \min\left\lbrace\frac{q-1}{2},d_{r},\sqrt{\frac{(q-1)n}{4(m+3)}},\sqrt{\frac{n(d_{r}+1)}{4(m+3)}}\right\rbrace$,\\ $1\leq d_{r}\leq q$  }\\  \midrule
				Hermitian Function Field (\cite{2}) &$\sqrt{q}-1$& $q-1$&$\frac{(q-1)(\sqrt{q}-1)-1}{\sqrt{q}(\sqrt{q}-1)m+\sqrt{q}-1}$\\  \midrule 
				Hermitian Function Field (\cite{3}) & $\sqrt{q}$ &  $ q-\sqrt{q}+1$  & $  \frac{\sqrt{q}-3}{\sqrt{q}+m}(q-\sqrt{q}+1)+\frac{2}{\sqrt{q}+m}$\\  \midrule
				Generalized Hermitian Function Field (\cite{21}) & $\sum^{r}_{i=1}q^{i}$&  $q^{r+1}-q^{r} +q-2$  &$\frac{(\sum^{r}_{i=1}q^{i})(q^{r+1}-q^{r}+q-2)-q^{r-1}(q-1)-m}{\sum^{r}_{i=1}q^{i}+mq^{r-1}(q-1)} $\\\midrule
				Cyclotomic Function Field (\cite{3}) & $\sqrt{q}$ &  $\makecell{\frac{  \sqrt{q}^{d}+1  }{\sqrt{q}+1},\ d\ \text{odd}}$  & \makecell{$  \frac{\sqrt{q}}{\sqrt{q}+m} \frac{ \sqrt{q}^{d}+1}{\sqrt{q}+1}-\frac{ma+1}{\sqrt{q}+m},$\\$ a=\frac{\sqrt{q}^{d}-\sqrt{q}}{\sqrt{q}+1} d$}\\  \midrule
				Suzuki Function Field (\cite{15}) &$q-1$, $q=2^{2t+1}$  &  $q-1$, $q=2^{2t+1}$  &$\frac{(q-1)^{2}-m(q-1+2q_{0})}{q-1+m(q-1-2q_{0})}, q_{0}=2^{t}$\\\midrule
				\makecell{\textbf{ Narrow Ray Class Field  via} \\ \textbf{ Elliptic Function Field}\\ \textbf{(Example IV.11)}}  &$\sqrt{q}+2\sqrt[4]{q}$& 
				$\makecell{\frac{\sqrt{q}^{d}+1}{\sqrt{q}+1}, d\ \text{odd}}$
				&\makecell{$\frac{\sqrt{q}+2\sqrt[4]{q}}{\sqrt{q}+2\sqrt[4]{q}+m}\frac{\sqrt{q}^{d}+1}{\sqrt{q}+1}
					-\frac{ma+1}{\sqrt{q}+2\sqrt[4]{q}+m}$,\\ $ a=\gamma\deg \mathcal{R}_{Q}$   }\\      \midrule
				\makecell{\textbf{Narrow Ray Class Field via }\\\textbf{ Hyperelliptic Function Field}\\\textbf{(Example IV.12)}} &\makecell{$\sqrt{q}+(t-1)\sqrt[4]{q}$,\\ $t\mid(\sqrt{q}-1)$}&\makecell{$\frac{\sqrt{q}^{d}+1}{\sqrt{q}+1}$, $ d\ \text{odd}$}
				& \makecell{$\frac{\sqrt{q}+(t-1)\sqrt[4]{q}}{\sqrt{q}+(t-1)\sqrt[4]{q}+m} \frac{\sqrt{q}^{d}+1}{\sqrt{q}+1}-\frac{ma+1}{\sqrt{q}+(t-1)\sqrt[4]{q}+m}$,\\ $ a=\gamma\deg \mathcal{R}_{Q}$  }\\   \midrule   
\makecell{\textbf{Narrow Ray Class Field via }\\\textbf{ Function Field $F/\mathbb{F}_{\sqrt{q}}$ with Genus $\mathfrak{g}_{F}\geq 1$}\\\textbf{(Theorem III.1)}} &\makecell{$\sqrt{q}+s_{\sqrt{q}}$,\\$ |s_{\sqrt{q}}|\leq\mathfrak{g}_{F}\cdot\lfloor2\sqrt[4]{q}\rfloor$}&\makecell{$\frac{\sqrt{q}^{d}+1}{\sqrt{q}+1}$, $ d\ \text{odd}$}
				& \makecell{$\frac{\sqrt{q}+ s_{\sqrt{q}} }{\sqrt{q}+s_{\sqrt{q}} +m} \frac{\sqrt{q}^{d}+1}{\sqrt{q}+1}-\frac{ma+1}{\sqrt{q}+s_{\sqrt{q}} +m}$,\\  $ a=\gamma\deg \mathcal{R}_{Q}$ }\\   \midrule
		\end{tabular}}
		\label{q}
	\end{table}

	\section{Multi-sequences with high nonlinear complexity  from Narrow Ray Class Fields  }
	\label{sec:3}
	In this section, we illustrate that Drinfeld modules of rank $1$ and narrow ray class fields provide useful tools to generate new multi-sequences with significantly high nonlinear complexity.
	
	\subsection{Construction of multi-sequences from Descent Narrow Ray Class Fields}

	Let $\mathbb{F}_{r}$ be a subfield of $\mathbb{F}_{q}, q=r^{2}$ and define two integral rings
	\begin{displaymath}
		\begin{split}
			A_{q}&=\{x\in F/\mathbb{F}_{q}: \mathrm{v}_{P}(x)\geq 0\ \text{for all}\ P\in\mathbb{P}_{F/\mathbb{F}_{q}}\ \text{with}\ P\neq\infty\}\\
			A_{r}&=\{x\in F/\mathbb{F}_{r}: \mathrm{v}_{P}(x)\geq 0\ \text{for all}\ P\in\mathbb{P}_{F/\mathbb{F}_{r}}\ \text{with}\ P\neq\infty\}.
		\end{split}
	\end{displaymath}
	Let $Q$ be a place of $F/\mathbb{F}_{r}$ of odd degree $d\geq 3$. Denote by the constant extension $L=\mathbb{F}_{q}\cdot F$.
    
	Then we have the following cyclic descent narrow ray class field.
	\begin{lemma}\label{lll} (Lemma 3.5, \cite{26})   (Cyclic descent narrow ray class field)
	Let $\mathbb{F}_{q}, q=r^{2}$ be an extension of $\mathbb{F}_{r}$. Fix a place $Q$ of $F/\mathbb{F}_{r}$ with odd degree $d\geq 3$. Then there exists a function field $H/\mathbb{F}_{q}$ such that:
    \begin{enumerate}
        \item $H/L$ is a cyclic Galois extension of degree
        \begin{displaymath}
            \mu_{H}:=\frac{r^{d}+1}{r+1};
        \end{displaymath}
        \item Each $\mathbb{F}_{r}-$rational place of $F$ completely splits in $H/L$.
        \item $Q$ is a unique finite place that can be ramified in $H/L$;
        \item We have an estimation of the genus
        \begin{displaymath}
            \mathfrak{g}_{H}\leq\mu_{H}(\mathfrak{g}_{F}-1)+\frac{d(\mu_{H}-1)}{2}+1.
        \end{displaymath}
    \end{enumerate}
	\end{lemma}
	
	\begin{remark}
    \begin{enumerate}
    \item  The cyclic descent narrow ray class field is actually the second descent. The first descent function field can be constructed as follows. Consider extension $\mathcal{H}_{A_{q}}(\Lambda_{Q})/\mathcal{H}_{A_{q}}$, then we have $\mathrm{Gal}(\mathcal{H}_{A_{q}}(\Lambda_{Q})/F/\mathbb{F}_{q})\simeq\mathrm{Cl}^{+}_{Q}(A_{q})$. Note that the extension is Galois and there exists a unique subfield $\Gamma$ fixed by $\mathbb{F}^{*}_{q}\cdot\mathrm{Cl}^{+}_{Q}(A_{r})$, \textit{that is}, $\Gamma=\{x\in \mathcal{H}_{A_{q}}(\Lambda_{Q}): \sigma(x)=x, \ \text{for}\ \sigma\ \in\mathbb{F}^{*}_{q}\cdot\mathrm{Cl}^{+}_{Q}(A_{r})\}$. By Galois theory, there is an isomorphism $\mathrm{Gal}(\mathcal{H}_{A_{q}}(\Lambda_{Q})/\Gamma)\simeq \mathbb{F}_{q}^{*}\cdot\mathrm{Cl}^{+}_{Q}(A_{r}).$
	The extension degree is given by
	\begin{displaymath}
		[\Gamma:L]=\frac{|\mathrm{Cl}^{+}_{Q}(A_{q})|}{|\mathbb{F}^{*}_{q}\cdot\mathrm{Cl}^{+}_{Q}(A_{r})|}=\frac{\mathrm{h}(F/\mathbb{F}_{q})(r^{d}+1)}{\mathrm{h}(F/\mathbb{F}_{r})(r+1)},
		\label{e1}
	\end{displaymath}
	where the last equality follows from 
	\begin{displaymath}
		\begin{split}
			& [\mathcal{H}_{A_{q}}(\Lambda_{\phi(Q)}):L]=|\mathrm{Cl}^{+}_{Q}(A_{q})|=\mathrm{h}(F/\mathbb{F}_{q})(r^{2d}-1),\\
			& |\mathrm{Cl}^{+}_{Q}(A_{r})\cdot\mathbb{F}_{q}^{*}|=|\mathrm{Cl}^{+}_{Q}(A_{r})|\cdot|\mathbb{F}_{q}^{*}|/|\mathrm{Cl}^{+}_{Q}\mathrm(A_{r})\cap\mathbb{F}_{q}^{*}|=\mathrm{h}(F/\mathbb{F}_{r})(r^{d}-1)(r+1).\\
		\end{split}
	\end{displaymath}
    For details, the readers may refer to \cite{49,26}.
\item  The extension degree of the first descent function field $\Gamma$ is given by
        \begin{displaymath}
            [\Gamma:L]=\frac{\mathrm{h}(F/\mathbb{F}_{q})(r^{d}+1)}{\mathrm{h}(F/\mathbb{F}_{r})(r+1)}
        \end{displaymath}
but $\Gamma/L$ is not necessarily cyclic. Hence, one cannot choose a cyclic automorphism of order $[\Gamma:L]$. The class number factor must not be included in the period by Galois theory.
\item In particular, if we take $F/\mathbb{F}_{r}$ as a rational function field, then the class numbers of $F/\mathbb{F}_{r}$ and $F/\mathbb{F}_{q}$ can be given by $\mathrm{h}(F/\mathbb{F}_{r})=\mathrm{h}(F/\mathbb{F}_{q})=1 $ and we have $$[\Gamma:L]=[H:L]=\frac{r^{d}+1}{r+1}, $$ which means that the first and second descents are identical. The genus of $H$ achieves equality, \textit{i.e.}, $$\mathfrak{g}_{H}=-\mu_{H}+\frac{d(\mu_{H}-1)}{2}+1=\frac{(d-1)(\mu_{H}-1)}{2}=\frac{(d-1)}{2}\left(\frac{r^{d}+1}{r+1}-1\right).$$ Another proof via the algebraic structure of cyclotomic function fields can be checked at Proposition IV.2 in \cite{100}.
\end{enumerate}
	\end{remark}

Let $\sigma$ be a generator of $\mathrm{G}=\mathrm{Gal}(H/L)$. Let $\mathcal{R}_{Q}=\sum_{R\mid Q}R$ be the reduced divisor supported on the place of $H$ lying above $Q$. It should be noted that $\mathcal{R}_{Q}$ is $\sigma-$invariant. Denote by $e_{Q}$ the common ramification index above $Q$, then 
    \begin{displaymath}
        \Delta_{Q}:=\deg\mathcal{R}_{Q}=\frac{\mu_{H}d}{e_{Q}}.
    \end{displaymath}
	In particular, we have $d\leq \Delta_{Q}\leq \mu_{H}d$, and $\Delta_{Q}=d$ in the totally ramified case.

	Let $\mathfrak{p}_{0},\mathfrak{p}_{1},\cdots,\mathfrak{p}_{r+s_{r}}$ be $r+s_{r}+1$ rational places in $F/\mathbb{F}_{r}$ and it is also the rational places in $L$. By Lemma \ref{lll}, each $\mathfrak{p}_{t}$ is completely splitting in $H$. Without loss of generality, choose any rational place $P_{t}$ lying over $\mathfrak{p}_{t}$ and denote 
	\begin{displaymath}
		\Theta_{t}=\{P_{t},\sigma(P_{t}),\dots,\sigma^{ \mu_{H} -1}(P_{t})\}.
	\end{displaymath}
    by the orbit of each $P_{t}$. In the following, $\Theta_{0}$ shall be used to generate a evaluation function while $\Theta_{t}:1\leq t\leq r+s_{r}$ will be used as evaluation places. 

 For our construction, we choose 
\begin{displaymath}
    A=\gamma\mathcal{R}_{Q}, a:=\deg A=\gamma\Delta_{Q},
\end{displaymath}
with $\gamma\geq 1$.  Note that $A$ is $\sigma-$invariant, whose support is disjoint from all rational orbits.

We impose the condition
\begin{displaymath}
    2\mathfrak{g}_{H}-1\leq a\leq (r+s_{r}-1)\mu_{H}-1.
\end{displaymath}

	The evaluation function is taken from the Riemann-Roch space $\mathcal{L}(A+P_{0})\setminus\mathcal{L}(A)$. The existence of $z$ is guaranteed by the Riemann-Roch theorem. Because $\deg A\geq 2\mathfrak{g}_{H}-1$ and $P_{0}\notin\mathrm{Supp}(A)$,
	\begin{displaymath}
   \ell(A+P_{0})-\ell(A)=1.
	\end{displaymath}
    Choose $0\neq z\in \mathcal{L}(A+P_{0})\setminus\mathcal{L}(A)$. Then $z$ has a simple pole at $P_0$ and all its other poles are contained in $\mathrm{Supp}(A)$.
	
	 Then we can construct a multi-sequence of dimension $r+s_{r}$ and period $\mu_{H}$:
	\begin{displaymath}
		\mathcal{S}=\{\mathbf{s}^{(t)}: 1\leq t\leq r+s_{r}\}
		\label{22}
	\end{displaymath}
	with each component given by $\mathbf{s}^{(t)}:=\{z(\sigma^{j}(P_{t})\}_{j=0}^{\infty}$ and $z\in\mathcal{L}(A+P_{0})\setminus\mathcal{L}(A)$.

	\subsection{High Linear Complexity}
	\begin{theorem}\label{t1}  Let $\mathbb{F}_{r}$ be a finite field with $r$ elements, and let $\mathbb{F}_{q}, q=r^{2}$ be an extension of $\mathbb{F}_{r}$.    Consider a function field $F/\mathbb{F}_{r}$ with genus $\mathfrak{g}_{F}$ and $1+r+s_{r}, |s_{r}|\leq \mathfrak{g}_{F}\cdot\lfloor 2\sqrt{r}\rfloor$ rational places. Fix a place $Q\in\mathbb{P}_{F}$ of odd degree $d\geq 3$ and $ 2\mathfrak{g}_{H}-1  \leq a\leq (r+s_{r}-1) \mu_{H} -1$.  Then we can construct a multi-sequence $$\mathcal{S}=\{\mathbf{s}^{(t)}:1\leq t\leq r+s_{r}\}$$ with dimension $r+s_{r}$, period $\mu_{H} $, and linear complexity 
		\begin{displaymath}
			\mathrm{L}(\mathcal{S})= \mu_{H}.
		\end{displaymath}
	\end{theorem}
	\begin{proof}
		Suppose that the linear complexity of $\mathcal{S}$ is strictly smaller than $\mu_{H}$. Let $\ell:=\ell(\mathcal{S})$. By the definition of linear complexity, there exist $\lambda_{0},\lambda_{1},\cdots,\lambda_{\ell}\in\mathbb{F}_{q}$ with $\lambda_{\ell}\neq 0$ such that
		\begin{displaymath}
	\sum^{\ell}_{i=0}\lambda_{i}s^{(u)}_{i+j}=\sum^{\ell}_{i=0}\lambda_{i}z(\sigma^{i+j}(P_{u}))=0
		\end{displaymath}
		for all $j\geq 0$ and any $1\leq u\leq r+s_{r}$.
		
		Now consider the function $w_{\ell}=\sum^{\ell}_{i=0}\lambda_{i}\sigma^{-i}(z)$.   We claim that the function $w_{\ell}$ is nonzero due to the fact that $\sigma^{-\ell}(z)$ has pole $\sigma^{-\ell}(P_{0})$ and $\lambda_{\ell}\neq 0$, while $\sigma^{-\ell}(P_{0})$ is not a pole for any $\sigma^{-i}(z), 0\leq i\leq \ell-1$.
        
         By Lemma \ref{l1}, we have $w_{\ell}\in \mathcal{L}(A+\sum^{\ell}_{i=0}\sigma^{-i}(P_{0}))$ from  $ (\sigma^{-i}(z))_{\infty}\leq A+ \sigma^{-i}(P_{0})$. It is also checked that  $$\sum^{\ell}_{j=0}\lambda_{i}\sigma^{-i}(z)\in\mathcal{L}\left(A+\sum^{\ell}_{i=0}\sigma^{-i}(P_{0})-\sum^{r+s_{r}}_{i=1}\sum^{ \mu_{H} -1}_{j=0}\sigma^{j}(P_{i})\right)$$ and we have $$\deg\left(A+\sum^{\ell}_{i=0}\sigma^{-i}(P_{0})-\sum^{r+s_{r}}_{i=1}\sum^{\mu_{H}-1}_{j=0}\sigma^{j}(P_{i})\right)\geq 0,$$ which is equivalent to $a+\ell+1\geq(r+s_{r}) \mu_{H}  $. However, it is impossible for $2\mathfrak{g}_{H}-1\leq a\leq (r+s_{r}-1) \mu_{H} -1$ and $1\leq \ell<  \mu_{H}$. Then we have the desired result.
	\end{proof}

    \begin{remark}
        If the base field is a rational function field, then we have cyclotomic function field; see Theorem 6 in \cite{3}. The genus of the function field $H$ is given by $\mathfrak{g}_{H}=\frac{\sqrt{q}(\sqrt{q}-1)}{2}$ and $\sqrt[2]{q^{3}}+1$ rational places if we choose $\deg P=3$. As demonstrated in \cite{17}, $H$ is isomorphic to the Hermitian function field.
    \end{remark}
	
	\begin{corollary}
		\begin{enumerate}
			\item When $F/\mathbb{F}_{\sqrt{q}}$ is taken as a rational function field and $\gamma=\frac{\sqrt{q}^{d}-\sqrt{q}}{\sqrt{q}+1}$, one has a multi-sequence $\mathcal{S}=\{\mathbf{s}_{t}: 1\leq t\leq \sqrt{q}\}$ with $  \mathbf{s}_{t}=\{z(\sigma^{i}(P_{t}))\}_{i=0}^{\frac{\sqrt{q}^{d}+1}{\sqrt{q}+1}-1}$, $1\leq t\leq \sqrt{q}$; see \cite{3}. Then the linear complexity of $\mathcal{S}$ is given by $\mathrm{L}(\mathcal{S})=  \frac{\sqrt{q}^{d}+1}{\sqrt{q}+1}  $.
			\item Under the conditions above, we have known that the fixed field $H$ is a Hermitian function field if $d=3$; see \cite{3}. The multi-sequence is modified to $\mathcal{S}=\{\mathbf{s}_{t}: 1\leq t\leq \sqrt{q}\}$ with $  \mathcal{S}_{t}=\{z(\sigma^{i}(P_{t}))\}_{i=0}^{q-\sqrt{q}}, 1\leq t\leq \sqrt{q}$. The linear complexity of $\mathcal{S}$ is given by $\mathrm{L}(\mathcal{S})=  q-\sqrt{q}+1  $.
		\end{enumerate}
		
	\end{corollary}
	
	\begin{corollary}(Subfamilies)
		Let $\alpha\in\mathbb{N}$ satisfy
		\begin{displaymath}
			\alpha\geq 1+(a+1)/\mu_{H}.
		\end{displaymath}
		Then any $\alpha$ periodic sequences taken from the multi-sequence $\mathcal{S}$ form a new multi-sequence of linear complexity $\mu_{H}$. In particular, we denote by $$\alpha_{0}=1+\lceil (a+1)/\mu_{H}\rceil,$$
         where $\lceil x\rceil$ denotes the ceiling of $x\in\mathbb{R}$
	\end{corollary}

	\subsection{High Nonlinear Complexity }
	
	\begin{theorem}  \label{1234}
		Let $\mathbb{F}_{r}$ be a finite field with $r$ elements, and let $\mathbb{F}_{q}, q=r^{2}$ be an extension of $\mathbb{F}_{r}$.    Consider a function field $F/\mathbb{F}_{r}$ with genus $\mathfrak{g}_{F}$ and $1+r+s_{r}, |s_{r}|\leq \mathfrak{g}_{F}\cdot\lfloor 2\sqrt{r}\rfloor$ rational places. Fix a place $Q\in\mathbb{P}_{F}$ of odd degree $d\geq 3$ and $ 2\mathfrak{g}_{H}-1  \leq a\leq (r+s_{r}-1) \mu_{H} -1$.  Then we can construct a multi-sequence $$\mathcal{S}=\{\mathbf{s}^{(t)}:1\leq t\leq r+s_{r}\}$$ with dimension $r+s_{r}$, period $ \mu_{H} $, and nonlinear complexity lower bounded by
		\begin{displaymath}
			\mathrm{NL}_{m}(\mathcal{S})\geq\frac{(r+s_{r})\mu_{H}-ma-1}{m+r+s_{r}}.
		\end{displaymath}
	\end{theorem}
	\begin{proof}
		First, we show that the sequence $\mathcal{S}$ is nonzero. In fact, if $\mathcal{S}=\mathbf{0}$, then there exists $z\in\mathcal{L}\big(A+P_{0}-\sum^{r+s_{r}}_{i=1}\sum^{ \mu_{H}-1}_{j=0}\sigma^{-j}(P_{i})\big)$. This condition implies that
		$$
		\deg\left(A+P_{0}-\sum^{r+s_{r}}_{i=1}\sum^{\mu_{H}  -1}_{j=0}\sigma^{-j}(P_{i})\right)\geq 0,
		$$
		that is, $a+1- (r+s_{r}) \mu_{H} \geq 0$, which leads to a contradiction, \textit{i.e.}, $2\mathfrak{g}_{H}-1\leq a\leq (r+s_{r}-1) \mu_{H} -1$.
		
		Let $\Phi(x_{1},\ldots,x_{\mathfrak{N}})=\sum_{\mathrm{wt}(I)\leq m}\phi_{I}x_{1}^{e_{1}}\cdots x_{\mathfrak{N}}^{e_{\mathfrak{N}}}$ be a nonzero polynomial in $\mathbb{F}_{q}[x_{1},\ldots,x_{\mathfrak{N}}]$ with a total degree not exceeding $m$ that generates the sequence $\mathcal{S}$. 
		
		If $\mathfrak{N}\geq  \mu_{H} =\mathrm{L}(\mathcal{S})$, then the bound is trivially valid. We now focus on the cases $\mathfrak{N}\leq \mu_{H} -1$. By definition, we have
		\begin{displaymath}
s^{(t)}_{j+\mathfrak{N}}=\Phi(s^{(t)}_{j},\cdots,s^{(t)}_{j+\mathfrak{N}-1})
			\label{55}
		\end{displaymath}
		for each $0\leq j\leq \mu_{H}-1-\mathfrak{N}$. It is equivalent to
		\begin{displaymath}
			\begin{split}
				z(\sigma^{j+\mathfrak{N}}(P_{t}))=\Phi(z(\sigma^{j}(P_{t})),&z(\sigma^{j+1}(P_{t})),\\
				&\cdots,
				z(\sigma^{\mathfrak{N}+j-1}(P_{t}))),\ z\in \mathcal{L}(( A+P_{0})\setminus\mathcal{L}(A)
			\end{split}
		\end{displaymath}
		for all $1\leq t\leq r+s_{r}$ and $0\leq j\leq \mu_{H}  -\mathfrak{N}-1$.
		Consider the function
		\begin{displaymath}
			f=\sigma^{-\mathfrak{N}}(z)-\Phi(z,\sigma^{-1}(z),\dots,\sigma^{-\mathfrak{N}+1}(z)).
		\end{displaymath}
		Given that $\sigma^{-i}(P_{0})$ is a pole of $\sigma^{-i}(z)$, for each $i\geq 0$. $\sigma^{-\mathfrak{N}}(P_{0})$ serves as a pole of $\sigma^{\mathfrak{-N}}(z)$, but not a pole of $\Phi(z,\sigma^{-1}(z),\dots,\sigma^{-\mathfrak{N}+1}(z))$. It follows that $f$ is a nonzero function. 
		
		Furthermore, from the analysis of the orders of poles, we have
		\[
		f\in\mathcal{L}\bigl(mA+\sigma^{-\mathfrak{N}}(P_{0})+\textstyle\sum_{i=0}^{\mathfrak{N}-1}m\sigma^{-i}(P_{0})\bigr).
		\]

		Then $f(\sigma^{j}(P_{t}))=0$ holds for all $1\leq t\leq r+s_{r}$ and $0\leq j\leq  \mu_{H} -\mathfrak{N}-1$. Thus $f$ has at least $(r+s_{r})( \mu_{H}  -\mathfrak{N})$ zeros. Therefore, we have
		\begin{displaymath}
			\deg\left(mA+\sigma^{\mathfrak{-N}}(P_{0})+\sum_{i=0}^{\mathfrak{N}-1}m\sigma^{-i}(P_{0})-\sum_{t=1}^{r+s_{r}}\sum_{j=0}^{\mu_{H}   -\mathfrak{N}-1}\sigma^{j}(P_{t})\right)\geq 0.
		\end{displaymath}
		Then we have
		\begin{displaymath}
			am+1+\mathfrak{N}m-   (r+s_{r})( \mu_{H}  -\mathfrak{N})\geq 0\implies 
			\mathfrak{N}\geq \frac{(r+s_{r}) \mu_{H}  -am-1}{m+r+s_{r}}
		\end{displaymath}
		and we have the desired result.
	\end{proof}

	\begin{corollary}
		Let $\mathbb{F}_{r}$ be a finite field with $r$ elements, and let $\mathbb{F}_{q}, q=r^{2}$ be an extension of $\mathbb{F}_{r}$.    Consider a function field $F/\mathbb{F}_{r}$ with genus $\mathfrak{g}_{F}$ and $1+r+s_{r}, |s_{r}|\leq \mathfrak{g}_{F}\cdot\lfloor 2\sqrt{r}\rfloor$ rational places. Fix a place $Q\in\mathbb{P}_{F}$ of an odd degree $d\geq 3$ and $ 2\mathfrak{g}_{H}-1  \leq a\leq (\alpha-1) \mu_{H} -1$.
           Then, for any $\alpha\geq \alpha_{0}$, we can construct a truncated multi-sequence $$\mathcal{S}=\{\mathbf{s}^{(t)}:1\leq t\leq \alpha\}$$ with dimension $\alpha$, period $\mu_{H}$, linear complexity $\mathrm{LC}(\mathcal{S})=\mu_{H}$ and nonlinear complexity lower bounded by
		\begin{displaymath}
			\mathrm{NL}_{m}(\mathcal{S},\alpha)\geq \frac{\mu_{H}\alpha-ma-1}{  m+\alpha },\ \alpha\geq\alpha_{0}.
		\end{displaymath}
	\end{corollary}

	\subsection{ High $\beta-$error Nonlinear Complexity}
	\begin{theorem}
		Let $\mathbb{F}_{r}$ be a finite field with $r$ elements, and let $\mathbb{F}_{q},q=r^{2}$ be an extension of $\mathbb{F}_{r}$.    Consider a function field $F/\mathbb{F}_{r}$ with genus $\mathfrak{g}_{F}$ and $1+r+s_{r}, |s_{r}|\leq \mathfrak{g}_{F}\cdot\lfloor 2\sqrt{r}\rfloor$ rational places. Fix a place $Q\in\mathbb{P}_{F}$ of odd degree $d\geq 3$ and $ 2\mathfrak{g}_{H}-1  \leq a\leq (r+s_{r}-1) \mu_{H} +\alpha_{0}-3$.   Then we can construct a multi-sequence $$\mathcal{S}=\{\mathbf{s}^{(t)}:1\leq t\leq r+s_{r}\}$$ with dimension $r+s_{r}$, period $\mu_{H}$, and we have the following results about $\beta-$error nonlinear complexity:
		
		$1)$ If $\beta\leq r+s_{r}-\alpha_{0}$,
		\begin{displaymath}
			\mathrm{NL}_{m,\beta}(\mathcal{S})\geq \frac{\mu_{H}(r+s_{r}-\beta)-ma-1}{m+ r+s_{r}-\beta   };
		\end{displaymath}

		$2)$ If $\beta\leq r+s_{r}-\alpha_{0}+1$,   $$  \mathrm{NL}_{m,\beta}(\mathcal{S})\geq   -1+\frac{\mu_{H}(r+s_{r})-ma+m-1}{ r+s_{r}-\alpha_{0}+1+m   };$$
		
		$3)$ If $\beta\leq r+s_{r}-\alpha_{0}+2$,  $$  \mathrm{NL}_{m,\beta}(\mathcal{S})\geq  -1+\min\left\{\frac{\mu_{H}(r+s_{r}-1)-ma+m-1}{ r+s_{r}-\alpha_{0}+m   },    \frac{\mu_{H}(r+s_{r})-ma+m-1}{ r+s_{r}-\alpha_{0}+2+m   }  \right\}.  $$
	\end{theorem}
	
	\begin{proof}
		$1)$ Suppose that multi-sequence $\mathcal{S}_{1}$ is obtained from changing $\beta$ components in $\mathcal{S}$. It is obvious that at most $\beta$ sequences are affected in $\mathcal{S}$ when transmitting. Then we have $|\mathcal{S}\cap\mathcal{S}_{1}|\geq r+s_{r}-\beta\geq \alpha_{0}$. By the corollary, it follows
		\begin{displaymath}
			\mathrm{NL}_{m,\beta}(\mathcal{S})\geq \mathrm{NL}_{m,\beta}(\mathcal{S}\cap\mathcal{S}_{1})\geq\frac{\mu_{H}(r+s_{r}-\beta)-ma-1}{m+ r+s_{r}-\beta   }.
		\end{displaymath}
		Since $U$ was arbitrary, taking the minimum over all $U$ with $d_{H}(\mathcal{S},U)\leq\beta$ gives the result.
		
		$2)$ Suppose that multi-sequence $\mathcal{S}_{2}$ is obtained from changing $\beta$ components in $\mathcal{S}$. Then $|\mathcal{S}\cap\mathcal{S}_{2}|\geq \alpha_{0}-1$. Otherwise, if $|\mathcal{S}\cap\mathcal{S}_{2}|\geq \alpha_{0}$, we have the first result. Now we assume $|\mathcal{S}\cap\mathcal{S}_{2}|=\alpha_{0}-1 $ and it is equivalent that there are $r+s_{r}-\alpha_{0}+1$ sequences with exactly one component modified in $\mathcal{S}$ when transmitting through noise channels. For our construction, we can assume $ \mathcal{S}\cap\mathcal{S}_{2}=\{\mathbf{s}_{t}: 1\leq t\leq \alpha_{0}-1\}$ and denote modified components by $z(\sigma ^{\eta_{t}-1}(P_{t}))$ for $t=\alpha_{0},\cdots,r+s_{r}$ and $1\leq \eta_{t}\leq \mu_{H}$. Denote by $\mathfrak{N}_{2}$ the nonlinear complexity of $\mathcal{S}_{2}$. Then for any $t=1,2,\cdots,\alpha_{0}-1$ and $j\geq 0$, we have 
		\begin{displaymath}
			\begin{split}
				z(\sigma^{j+\mathfrak{N}_{2}}(P_{t}))=\Phi(z(\sigma^{j}(P_{t})),&z(\sigma^{j+1}(P_{t})),\\
				&\cdots,
				z(\sigma^{\mathfrak{N}_{2}+j-1}(P_{t}))),\ z\in \mathcal{L}(A+P_{0})\setminus\mathcal{L}(A).
			\end{split}
		\end{displaymath}
		For any $t=\alpha_{0},\alpha_{0}+1,\cdots,r+s_{r}$ and $\eta_{t}+1\leq j\leq \mu_{H}+\eta_{t}-1-\mathfrak{N}_{2}$, we also have 
		\begin{displaymath}
			\begin{split}
				z(\sigma^{j+\mathfrak{N}_{2}}(P_{t}))=\Phi(z(\sigma^{j}(P_{t})),&z(\sigma^{j+1}(P_{t})),\\
				&\cdots,
				z(\sigma^{\mathfrak{N}_{2}+j-1}(P_{t}))),\ z\in \mathcal{L}(A+P_{0})\setminus\mathcal{L}(A).
			\end{split}
		\end{displaymath}
		It suffices to consider $f=\sigma^{\mathfrak{-N}_{2}}(z)-\Phi(z,\sigma^{-1}(z),\cdots,\sigma^{-\mathfrak{N}_{2}+1}(z))$. It is obvious
		\begin{small}
			\begin{displaymath}
				f\in\mathcal{L}\bigl(mA+\sigma^{-\mathfrak{N}_{2}}(P_{0})+\textstyle\sum_{i=0}^{\mathfrak{N}_{2}-1}m\sigma^{-i}(P_{0})-\sum^{\alpha_{0}-1}_{t=1}\sum^{\mu_{H}-1}_{i=0}\sigma^{i}(P_{t})-\sum^{r+s_{r}}_{t=\alpha_{0}}\sum^{\mu_{H}+\eta_{t}-\mathfrak{N}_{2}-1}_{i=\eta_{t}+1}\sigma^{i}(P_{t})\bigr).
			\end{displaymath}
		\end{small}
		Then we have
		\begin{small}
			\begin{displaymath}
				\deg(mA+\sigma^{-\mathfrak{N}_{2}}(P_{0})+\textstyle\sum_{i=0}^{\mathfrak{N}_{2}-1}m\sigma^{-i}(P_{0})-\sum^{\alpha_{0}-1}_{t=1}\sum^{\mu_{H}-1}_{i=0}\sigma^{i}(P_{t})-\sum^{r+s_{r}}_{t=\alpha_{0}}\sum^{\mu_{H}+\eta_{t}-\mathfrak{N}_{2}-1}_{i=\eta_{t}+1}\sigma^{i}(P_{t}))\geq 0
			\end{displaymath}
		\end{small}
		and $ma+1+m\mathfrak{N}_{2}-(\alpha_{0}-1)\mu_{H}-(r+s_{r}-\alpha_{0}+1)(\mu_{H}-\mathfrak{N}_{2}-1)\geq 0$ \textit{i.e.}
		\begin{displaymath}
			\mathfrak{N}_{2}\geq -1+\frac{\mu_{H}(r+s_{r})-ma+m-1}{ r+s_{r}-\alpha_{0}+1+m   }.
		\end{displaymath}
		Let $\mathcal{S}_{2}$ be transversal through all multi-sequences in $\mathcal{S}$, then the claim holds for any $\beta\leq r+s_{r}-\alpha_{0}$.
		
		$3)$ Suppose that multi-sequence $\mathcal{S}_{3}$ is obtained from changing $\beta=r+s_{r}-\alpha_{0}+2$ components in $\mathcal{S}$. Then $|\mathcal{S}\cap\mathcal{S}_{3}|\geq \alpha_{0}-2$. Denote by $\mathfrak{N}_{3}$ the nonlinear complexity of $\mathcal{S}_{3}$.
		
		First, if $|\mathcal{S}\cap\mathcal{S}_{3}|\geq \alpha_{0}$, we have the first result.
In 
addition, if $|\mathcal{S}\cap\mathcal{S}_{3}|=\alpha_{0}-1 $, then there are $r+s_{r}-\alpha_{0}+1$ modified sequences, which means that there exists a unique sequence with $2$ modified components. By similar discussions above, the function $f=\sigma^{-\mathfrak{N}_{3}}(z)-\Phi(z,\sigma^{-1}(z),\cdots,\sigma^{-\mathfrak{N}_{3}+1}(z))$ is in the Riemann-Roch space
		\begin{displaymath}
			\mathcal{L}\bigl(mA+\sigma^{-\mathfrak{N}_{3}}(P_{0})+\textstyle\sum_{i=0}^{\mathfrak{N}_{3}-1}m\sigma^{-i}(P_{0})-\sum^{\alpha_{0}-1}_{t=1}\sum^{\mu_{H}-1}_{i=0}\sigma^{i}(P_{t})-\sum^{r+s_{r}+1}_{t=\alpha_{0}}\sum^{\mu_{H}+\eta_{t}-\mathfrak{N}_{3}-1}_{i=\eta_{t}+1}\sigma^{i}(P_{t})\bigr).
		\end{displaymath}
		By \begin{small}$$\deg\left( mA+\sigma^{-\mathfrak{N}_{3}}(P_{0})+\textstyle\sum_{i=0}^{\mathfrak{N}_{3}-1}m\sigma^{-i}(P_{0})-\sum^{\alpha_{0}-1}_{t=1}\sum^{\mu_{H}-1}_{i=0}\sigma^{i}(P_{t})-\sum^{r+s_{r}+1}_{t=\alpha_{0}}\sum^{\mu_{H}+\eta_{t}-\mathfrak{N}_{3}-1}_{i=\eta_{t}+1}\sigma^{i}(P_{t}) \right)\geq 0  $$\end{small}, we have
		\begin{displaymath}
			\mathfrak{N}_{3}\geq  -1+\frac{\mu_{H}(r+s_{r}-1)-ma+m-1}{ r+s_{r}-\alpha_{0}+m   } .
		\end{displaymath}
		
		Finally, if $|\mathcal{S}\cap\mathcal{S}_{3}|= \alpha_{0}-2$, there are $r+s_{r}-\alpha_{0}+2$ sequences with exactly one component modified in $\mathcal{S}$ when they are transmitted in noise channels. By similar discussions as before, $f=\sigma^{-\mathfrak{N}_{3}}(z)-\Phi(z,\sigma^{-1}(z),\cdots,\sigma^{-\mathfrak{N}_{3}+1}(z))$ lies in the Riemann-Roch space
		\begin{displaymath}
			\mathcal{L}\bigl(mA+\sigma^{-\mathfrak{N}_{3}}(P_{0})+\textstyle\sum_{i=0}^{\mathfrak{N}_{3}-1}m\sigma^{-i}(P_{0})-\sum^{\alpha_{0}-2}_{t=1}\sum^{\mu_{H}-1}_{i=0}\sigma^{i}(P_{t})-\sum^{r+s_{r}+2}_{t=\alpha_{0}}\sum^{\mu_{H}+\eta_{t}-\mathfrak{N}_{3}-1}_{i=\eta_{t}+1}\sigma^{i}(P_{t})\bigr)
		\end{displaymath}
		and we have 
		$mA+1+m\mathfrak{N}_{3}-(\alpha_{0}-2)\mu_{H}-(r+s_{r}-\alpha_{0}+2)(\mu_{H}-\mathfrak{N}_{3}-1)\geq 0$ \textit{i.e.}
		\begin{displaymath}
			\mathfrak{N}_{3}\geq-1+ \frac{\mu_{H}(r+s_{r})-ma+m-1}{ r+s_{r}-\alpha_{0}+2+m   }.
		\end{displaymath}
		Then we have 
		\begin{displaymath}
			\begin{split}
				&\mathrm{NL}_{m,\beta}(\mathcal{S}_{3})=\mathrm{NL}_{m}(\mathcal{S}_{3})\geq\\
				&-1+\min\left\{\frac{\mu_{H}(r+s_{r}-1)-ma+m-1}{ r+s_{r}-\alpha_{0}+m   },    \frac{\mu_{H}(r+s_{r})-ma+m-1}{ r+s_{r}-\alpha_{0}+2+m   }  \right\}
			\end{split}
		\end{displaymath}
		for all $\beta\leq r+s_{r}-\alpha_{0}+2$.
	\end{proof}

	\begin{example}(Cyclotomic Function Field)
    Consider the Carlitz module defined by $\psi_{T}(u)=u^{q}+Tu$ with $u\in\overline{\mathbb{F}_{q}(T)}$, and we take an irreducible polynomial $p(T)$ of odd degree $d$ over $\mathbb{F}_{\sqrt{q}}$. Let $K=\mathbb{F}_{q}(T)$ be a rational function field. Then we have a cyclotomic function field $K(\Lambda_{p(T)})$ (see \cite{3}). Furthermore, we obtained the descent function field $H$ which is fixed by $(\mathbb{F}_{\sqrt{q}}[T]/(p(T))^{*}\cdot\mathbb{F}_{q}^{*}$. Then the nonlinear complexity of $\mathcal{S}$ is lower bounded by
			\begin{displaymath}
				\mathrm{NL}_{m}(\mathcal{S})\geq\frac{ \sqrt{q}}{m+\sqrt{q}}\times\frac{\sqrt{q}^{d}+1}{\sqrt{q}+1} - \frac{ ma+1}{m+\sqrt{q}},
			\end{displaymath}
where $a=\deg A=\gamma d$ in this case.
	\end{example}
    \begin{remark}
        The nonlinear complexity of $\mathcal{S}$ given in \cite{3} is lower bounded by
\begin{displaymath}
				\mathrm{NL}_{m}(\mathcal{S})\geq\frac{ \sqrt{q}-d}{m+\sqrt{q}}\times\frac{\sqrt{q}^{d}+1}{\sqrt{q}+1} + \frac{d-1}{m+\sqrt{q}}
			\end{displaymath}
        and it is  derived from the Riemann-Roch space $$\mathcal{L}\left((\mu_{H}-1)Q+\sigma^{\mathfrak{n}_{3}}(P_{0})+\sum^{\mathfrak{N}_{3}-1}_{i=1}m\sigma^{i}(P_{0})-\sum^{\sqrt{q}}_{i=1}\sum^{\mu_{H}-\mathfrak{N}_{3}-1}_{j=0}\sigma^{-j}(P_{i})\right).$$ Under the total degree convention used in this paper, the corresponding pole estimate becomes $$\mathcal{L}\left(m(\mu_{H}-1)Q+\sigma^{\mathfrak{n}_{3}}(P_{0})+\sum^{\mathfrak{N}_{3}-1}_{i=1}m\sigma^{i}(P_{0})-\sum^{\sqrt{q}}_{i=1}\sum^{\mu_{H}-\mathfrak{N}_{3}-1}_{j=0}\sigma^{-j}(P_{i})\right).$$ 
    \end{remark}

\begin{example}\text{(Elliptic Function Field)}
		Consider a maximal elliptic function field $F/\mathbb{F}_{\sqrt{q}}$ with the number of $\mathbb{F}_{\sqrt{q}}-$rational places given by $\sqrt{q}+2\sqrt[4]{q} +1$, see \cite{229}.
        The genus $\mathfrak{g}_{H}$ is bounded by
        \begin{displaymath}
        \mathfrak{g}_{H}\leq \frac{d}{2}\left( \frac{ \sqrt{q} ^{d}-\sqrt{q}}{  \sqrt{q}+1}  \right)+1.
        \end{displaymath}

        Choose a divisor $A=\gamma \mathcal{R}_{Q}$ with odd degree $d\geq 3$. We can construct a multi-sequence $\mathcal{S}$  with nonlinear complexity lower bounded by
		\begin{small}
			\begin{displaymath}
				\mathrm{NL}_{m}(\mathcal{S})\geq \frac{  (\sqrt{q}+2 \sqrt[4]{q}  )  }{\sqrt{q}+2\sqrt[4]{q}  +m} \frac{ \sqrt{q}  ^{d}+1}{ \sqrt{q} +1} -  \frac{  ma+1}{\sqrt{q}+2\sqrt[4]{q}  +m} 
			\end{displaymath}
		\end{small}
		and $\beta-$error nonlinear complexity of $\mathcal{S}$ is divided by
		
		$1)$   If $\beta\leq  \sqrt{q}+2\sqrt[4]{q}  -\alpha_{0}$,
		\begin{small}
			\begin{displaymath}
				\mathrm{NL}_{m,\beta}(\mathcal{S})\geq \frac{\mu_{H}( \sqrt{q}+2\sqrt[4]{q}-\beta)  -ma-1}{m+  \sqrt{q}+2\sqrt[4]{q} -\beta };  
			\end{displaymath}
		\end{small}
		
		$2)$ If $\beta\leq \sqrt{q}+2\sqrt[4]{q}+1  -\alpha_{0}$
		\begin{small}
			\begin{displaymath}
				\mathrm{NL}_{m,\beta}(\mathcal{S})\geq -1+\frac{\mu_{H}(\sqrt{q}+2\sqrt[4]{q})-ma+m-1}{ \sqrt{q}+2\sqrt[4]{q}+1-\alpha_{0}+m   };  
			\end{displaymath}
		\end{small}
		
		$3)$ If $\beta\leq \sqrt{q}+2\sqrt[4]{q}+2-\alpha_{0}$,
		\begin{small}
			\begin{displaymath}
				\mathrm{NL}_{m,\beta}(\mathcal{S})\geq  -1+\min\left\{\frac{\mu_{H}(  \sqrt{q}+2\sqrt[4]{q}-1)-ma+m-1}{ \sqrt{q}+2\sqrt[4]{q} - \alpha_{0}  +m   },    \frac{\mu_{H}( \sqrt{q}+2\sqrt[4]{q})-ma+m-1}{ \sqrt{q}+2\sqrt[4]{q} +2- \alpha_{0} +m   }  \right\}.  
			\end{displaymath}
		\end{small}

	\end{example}

	\begin{example}\text{(Hyperelliptic Function Field)}
		Consider a hyperelliptic function field $F/\mathbb{F}_{\sqrt{q}}$ with the defining equation $$\mathcal{Y}:y^{2}=x^{t}+1.$$ Assume $t\mid \sqrt{q}+1$. Then $F$ is a maximal function field and has genus $\mathfrak{g}_{F}=\frac{t-1}{2}$ with $\sqrt{q}+(t-1)\sqrt[4]{q} +1$ $\mathbb{F}_{\sqrt{q}}-$rational places, see \cite{229}.
        The genus $\mathfrak{g}_{H}$ is bounded by
        \begin{displaymath}
        \mathfrak{g}_{H}\leq \frac{t-3}{2}\frac{ \sqrt{q} ^{d}+1}{  \sqrt{q}+1}+\frac{d}{2}\left( \frac{ \sqrt{q} ^{d}-\sqrt{q}}{  \sqrt{q}+1}  \right)+1.
        \end{displaymath}

        Choose a divisor $A=\gamma \mathcal{R}_{Q}$ with odd degree $d\geq 3$. We can construct a multi-sequence $\mathcal{S}$  with nonlinear complexity lower bounded by
		\begin{small}
			\begin{displaymath}
				\mathrm{NL}_{m}(\mathcal{S})\geq \frac{  (\sqrt{q}+(t-1) \sqrt[4]{q}  )  }{\sqrt{q}+(t-1)\sqrt[4]{q}  +m} \frac{ \sqrt{q}  ^{d}+1}{ \sqrt{q} +1} -  \frac{  ma+1}{\sqrt{q}+(t-1)\sqrt[4]{q}  +m} .
			\end{displaymath}
		\end{small}
		The $\beta-$error nonlinear complexity of $\mathcal{S}$ is divided by
		
		$1)$   If $\beta\leq  \sqrt{q}+(t-1)\sqrt[4]{q}  -\alpha_{0}$,
		\begin{small}
			\begin{displaymath}
				\mathrm{NL}_{m,\beta}(\mathcal{S})\geq \frac{\mu_{H}( \sqrt{q}+(t-1)\sqrt[4]{q}-\beta)  -ma-1}{m+  \sqrt{q}+(t-1)\sqrt[4]{q}-\beta  };  
			\end{displaymath}
		\end{small}
		
		$2)$ If $\beta\leq \sqrt{q}+(t-1)\sqrt[4]{q}+1  -\alpha_{0}$
		\begin{small}
			\begin{displaymath}
				\mathrm{NL}_{m,\beta}(\mathcal{S})\geq -1+\frac{\mu_{H}(\sqrt{q}+(t-1)\sqrt[4]{q})-ma+m-1}{ \sqrt{q}+(t-1)\sqrt[4]{q}+ 1-\alpha_{0} +m   };  
			\end{displaymath}
		\end{small}
		
		$3)$ If $\beta\leq \sqrt{q}+(t-1)\sqrt[4]{q}+2-\alpha_{0}$,
		\begin{small}
			\begin{displaymath}
				\mathrm{NL}_{m,\beta}(\mathcal{S})\geq  -1+\min\left\{\frac{\mu_{H}(  \sqrt{q}+(t-1)\sqrt[4]{q}-1)-ma+m-1}{ \sqrt{q}+(t-1)\sqrt[4]{q} -\alpha_{0}  +m   },    \frac{\mu_{H}( \sqrt{q}+(t-1)\sqrt[4]{q})-ma+m-1}{ \sqrt{q}+(t-1)\sqrt[4]{q} +2- \alpha_{0}  +m   }  \right\}.  
			\end{displaymath}
		\end{small}

	\end{example}

	\section{ Asymptotic families via function field towers}
	\label{sec:5}
	In this section, we provide some asymptotic parameters via function field towers and Ihara's constant. 
\begin{lemma}\label{l7} (Lemma 2.1 (iv), \cite{41})
    Let $F/\mathbb{F}_{r}$ be a function field with the genus $\mathfrak{g}_{F}$ and constant field $\mathbb{F}_{r}$. If there exists a positive integer $d$ with $$d\geq \log_{r}(j\log j+1)+\sqrt{\mathfrak{g}_{F}},$$ then the number of places with degree $d$ is at least $j$.
\end{lemma}

    \begin{theorem}
        Let $\mathcal{F}=(F_{i}/\mathbb{F}_{r})_{i\geq 1}$ be a function field tower with genus $\mathfrak{g}_{i}\to\infty$ and denote by $N_{i}=N(F_{i})$ the number of rational places of $F_{i}$. Assume the existence of the following ratio, \textit{i.e.}
        \begin{displaymath}
            \lambda(\mathcal{F})=\lim\inf_{i\to \infty}\frac{N_{i}}{\mathfrak{g}_{i}}>2.
        \end{displaymath}
        Then, for sufficiently large $i$, we can choose an odd degree $Q_{i}$ of degree $d_{i}=O(\sqrt{\mathfrak{g}_{i}})$ and construct a multi-sequence $\mathcal{S}_{i}$ over $\mathbb{F}_{q}$ with 
        dimension $N_{i}-1$, period $$\mu_{i}=\frac{r^{d_{i}}+1}{r+1},$$ linear complexity $\mu_{i}$. Moreover, for every fixed integer $m\geq 1$, the nonlinear complexity is lower bounded by
        \begin{displaymath}
            \lim\inf_{i\to\infty}\frac{\mathrm{NL}_{m}(\mathcal{S}_{i})}{\mu_{i}}\geq 1-\frac{2m}{\lambda(\mathcal{F})}.
        \end{displaymath}
        In particular, the nonlinear complexity is bounded away from zero whenever $\lambda(\mathcal{F})>2m$.
\begin{proof}
   For each suﬀiciently large $i$, let $d_{i}$ be the smallest odd integer that satisfies $d_{i}\geq \sqrt{\mathfrak{g}_{i}}$.
   Hence $d_{i}=O( \sqrt{\mathfrak{g}_{i}} )$. By Lemma \ref{l7}, $F_{i}$ has a place of degree $d_{i}$ for sufficiently large $i$. 

   Denote by the constant field extension $L_{i}= \mathbb{F}_{q}\cdot F_{i} $. Let $H_{i}/L_{i}$ be the cyclic narrow ray extension tower, and let
    \begin{displaymath}
        \mu_{i}=\frac{r^{d_{i}}+1}{r+1}.
    \end{displaymath}

Write $\mathcal{R}_{i}$ for the reduced divisor lying above $Q_{i}$ and $\Delta_{i}=\deg\mathcal{R}_{i}$. Since $\Delta_{i}=\mu_{i}d_{i}/e_{Q_{i}}$, we have $\Delta_{i}\leq \mu_{i}d_{i}$. Choose the minimum $\gamma_{i}$ such that
\begin{displaymath}
    a_{i}:=\gamma_{i}\Delta_{i}\geq 2\mathfrak{g}_{H_{i}}-1
\end{displaymath}
and
\begin{displaymath}
    a_{i}<2\mathfrak{g}_{H_{i}}-1+\Delta_{i}.
\end{displaymath}
By the Hurtwitz genus formula, we have
\begin{displaymath}
    2\mathfrak{g}_{H_{i}}-1\leq 2\mu_{i}(\mathfrak{g}_{i}-1)+d_{i}(\mu_{i}-1)+1
\end{displaymath}
together with
$$\Delta_{i}\leq \mu_{i}d_{i},$$
 and we obtain
\begin{displaymath}
    \frac{a_{i}}{\mu_{i}}\leq2\mathfrak{g}_{i}+2d_{i}+O(1).
\end{displaymath}
Consequently, the following upper bound holds
\begin{displaymath}
    \lim\sup_{i\to\infty}\frac{a_{i}}{\mu_{i}\mathfrak{g}_{i}}\leq 2.
\end{displaymath}

Let $b_{i}=N_{i}-1$. Since $\lambda(\mathcal{F})>2$, $\frac{a_{i}}{\mu_{i}}\leq 2\mathfrak{g}_{i}+2d_{i}+O(1)$ and $d_{i}=O(\sqrt{\mathfrak{g}_{i}})$. Then we have 
\begin{displaymath}
    a_{i}\leq(b_{i}-1)\mu_{i}-1,
\end{displaymath}
for sufficiently large $i$. Then by Theorem \ref{t1}, we have $\mathrm{LC}(\mathcal{S}_{i})=\mu_{i}$.

We also have the following bound for nonlinear complexity 
\begin{displaymath}
    \frac{\mathrm{NL}_{m}(\mathcal{S}_{i})}{\mu_{i}}\geq\frac{b_{i}-m(a_{i}/\mu_{i})-1/\mu_{i}}{m+b_{i}}.
\end{displaymath}
Together with $\lim\inf b_{i}/\mathfrak{g}_{i}=\lambda(\mathcal{F})$ and $\lim\sup a_{i}/(\mu_{i}\mathfrak{g}_{i})\leq 2$, we have the desired result.
\end{proof}
    \end{theorem}
	
	\begin{corollary}
	   Let $q$ be a fourth power and $\mathcal{F}=(F_{i}/\mathbb{F}_{\sqrt{q}})_{i\geq 1}$ be a Garcia-Stichtenoth tower, which reaches the Drinfeld-Vladut bound; therefore, we have
        \begin{displaymath}
            \lim_{i\to\infty}\frac{N(F_{i})}{\mathfrak{g}_{F_{i}}}=\sqrt[4]{q}-1
        \end{displaymath}
If $\sqrt[4]{q}-1>2$, then we have
\begin{displaymath}
    \frac{\mathrm{LC}(\mathcal{S}_{i})}{\mu_{i}}=1
\end{displaymath}
and, for each fixed $m\geq 1$, we have
       \begin{displaymath}
            \lim\inf_{i\to\infty}\frac{\mathrm{NL}_{m}(\mathcal{S}_{i})}{\mu_{i}}\geq 1-\frac{2m}{\sqrt[4]{q}-1}.
       \end{displaymath} 
    Then the lower bound of the normalized nonlinear complexity is positive whenever $\sqrt[4]{q}-1>2m$, and for fixed $m$, the bound tends to $1$ as $q\to\infty$ through fourth powers.
\end{corollary}

	\section{comparison and Conclusion}
	\label{sec:6}
	In this paper, we construct new multi-sequences with high nonlinear complexity via descent narrow ray class fields by using algebraic structures of rank $1$ Drinfeld modules. It turns out that the multi-sequence generated from the descent narrow ray class fields has period $\mu_{H}=\frac{r^{d}+1}{r+1}$, high dimension $r+s_{r}$, lower bound of nonlinear complexities $$ \mathrm{NL}_{m}(\mathcal{S})\geq  \frac{(r+s_{r})\mu_{H}-ma-1}{m+r+s_{r}}.$$ We also obtain some results of $\beta-$error nonlinear complexity and it is divided into following cases
    \begin{enumerate}
    	 \item If $\beta\leq r+s_{r}-\alpha_{0}$,
		\begin{displaymath}
			\mathrm{NL}_{m,\beta}(\mathcal{S})\geq \frac{\mu_{H}(r+s_{r}-\beta)-ma-1}{m+ r+s_{r}-\beta   };
		\end{displaymath}

		\item If $\beta\leq r+s_{r}-\alpha_{0}+1$, 
        \begin{displaymath}
        \mathrm{NL}_{m,\beta}(\mathcal{S})\geq   -1+\frac{\mu_{H}(r+s_{r})-ma+m-1}{ r+s_{r}-\alpha_{0}+1+m   };
		\end{displaymath}
	\item  If $\beta\leq r+s_{r}-\alpha_{0}+2$,  $$  \mathrm{NL}_{m,\beta}(\mathcal{S})\geq  -1+\min\left\{\frac{\mu_{H}(r+s_{r}-1)-ma+m-1}{ r+s_{r}-\alpha_{0}+m   },    \frac{\mu_{H}(r+s_{r})-ma+m-1}{ r+s_{r}-\alpha_{0}+2+m   }  \right\}.  $$
    \end{enumerate}

     In particular, some asymptotic parameters are exploited by using function field towers, especially for the Garcia-Stichtenoth tower:
    \begin{displaymath}
            \lim\inf_{i\to\infty}\frac{\mathrm{NL}_{m}(\mathcal{S}_{i})}{\mu_{i}}\geq 1-\frac{2m}{\sqrt[4]{q}-1}.
       \end{displaymath}  
       As for possible future work, it is interesting to construct new multi-sequences with high nonlinear complexity by employing more advanced mathematical tools from algebraic geometry.

\end{document}